\documentclass[a4paper, 10pt, conference]{ieeeconf}      %

\IEEEoverridecommandlockouts                              %

\usepackage[utf8]{inputenc} 

\usepackage{graphics} %
\usepackage{epsfig} %
\usepackage{amsmath} %
\usepackage{amssymb}  %
\usepackage{booktabs}
\usepackage[sort]{cite}
\usepackage{algorithm}%
\usepackage{algpseudocode}%
\usepackage{siunitx}
\usepackage{import} %
\usepackage{color}
\usepackage{transparent}
\usepackage{svg}
\usepackage{multirow}
\usepackage{tikz}
\definecolor{mplblue}{RGB}{31,119,180}
\newcommand*\circled[1]{\tikz[baseline=(char.base)]{\node[shape=circle,draw,inner sep=1pt, color=mplblue] (char) {\footnotesize #1};}}

\newtheorem{lemma}{Lemma}
\newtheorem{claim}{Claim}
\newtheorem{theorem}{Theorem}
\newtheorem{definition}{Definition}

\title{\LARGE \bf
PredicTor: Predictive Congestion Control for the Tor Network}

\author{Felix Fiedler$^{1,*}$ , Christoph Döpmann$^{2,*}$,  Florian Tschorsch$^{2}$ and Sergio Lucia$^{1}$%
\thanks{*The contribution of Felix Fiedler and Christoph Döpmann is equal.}%
\thanks{Felix Fiedler acknowledges the support of the Helmholtz Einstein
	International Berlin Research School in Data Science (HEIBRiDS), German
	Aerospace Center (DLR), and Technische Universit\"at Berlin.}%
\thanks{$^{1}$Chair of Internet of Things for Smart Buildings, Technische Universit\"at Berlin, Einsteinufer 17, 10587 Berlin, Germany}%
\thanks{$^{2}$Chair of Distributed Security Infrastructures, Technische Universit\"at Berlin, Einsteinufer 17, 10587 Berlin, Germany}%
}

\newcommand\copyrighttext{%
  \footnotesize \copyright{}~2020~IEEE.  Personal use of this material is permitted.  Permission from IEEE must be obtained for all other uses, in any current or future media, including reprinting/republishing this material for advertising or promotional purposes, creating new collective works, for resale or redistribution to servers or lists, or reuse of any copyrighted component of this work in other works.}
\newcommand\copyrightnotice{%
\begin{tikzpicture}[remember picture,overlay]
\node[anchor=south,yshift=10pt,text=gray] at (current page.south) {\parbox{\dimexpr\textwidth-\fboxsep-\fboxrule\relax}{\copyrighttext}};
\end{tikzpicture}%
}

\begin{document}

\maketitle
\thispagestyle{empty}
\pagestyle{empty}

\copyrightnotice

\begin{abstract}
In the Tor network,
anonymity is achieved through a multi-layered architecture,
which comes at the cost of a complex network. 
Scheduling data in this network is a challenging task and the current approach 
shows to be incapable of avoiding network congestion and allocating fair data rates.
We propose PredicTor, a distributed model predictive control approach, to tackle these challenges.
PredicTor is designed to schedule incoming and outgoing data rates on individual nodes of the Tor architecture, leading to a scalable approach.
We successfully avoid congestion through exchanging information of predicted behavior with adjacent nodes.
Furthermore, we formulate PredicTor with a focus on fair allocation of resources, 
for which we present and proof a novel optimization-based fairness approach.
Our proposed controller is evaluated with the popular network simulator ns-3,
where we compare it with the current Tor scheduler as well as with another recently proposed enhancement.
PredicTor shows significant improvements over the previous approaches, especially with respect to latency.
\end{abstract}

\section{Introduction}
The Tor network allows its users to anonymously access the Internet,
and thus serves an important societal role by supporting freedom of press and speech.
It consists of an overlay network connecting so-called relay nodes,
which can be used to establish anonymous connections.
To this end, the Tor client software builds a cryptographically-secured \emph{circuit},
a path over three relays, where each relay knows its immediate neighbors only.

While an extra delay is inevitable to gain anonymity (due to re-routing the traffic),
the performance---in terms of latency, data rates, and fairness---is however
neither optimal nor stable~\cite{reardon2008improving,DBLP:conf/p2p/DhungelSRHR10}.
One of the major shortcomings is the lack of fair rate allocation~\cite{Tschorsch2011}
and an effective congestion control~\cite{alsabah2011defenestrator,DBLP:conf/p2p/DhungelSRHR10}.
Here, \emph{congestion control} describes the nontrivial task of
scheduling data transmissions in a way
that minimizes network load while obtaining the maximum possible throughput.
Relaying data over a series of nodes, like in Tor, amplifies the problem;
especially when rising delays occur in the network.
In particular, Tor relays are unable to react to congestion,
for example by signaling upstream to throttle sending rates.

\begin{figure}
	\scriptsize
	\centering
	\def\svgwidth{1\linewidth}
	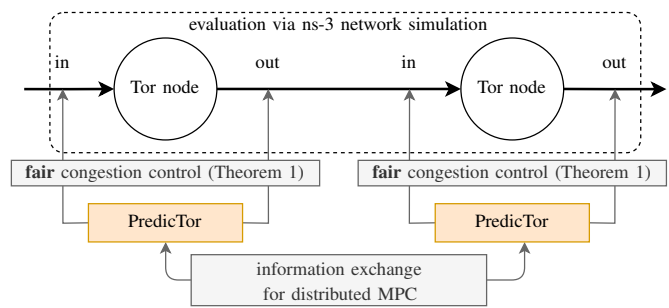
	\caption{Overview of the proposed method.}
	\label{fig:paper_overview}
\end{figure}

Different methods have been proposed
to improve the performance of the Tor network,
including the adaptation of standard congestion control algorithms to Tor~\cite{reardon2008improving,alsabah2013pctcp},
as well as the development of tailored approaches~\cite{backtap,alsabah2011defenestrator}.
Most notably, PCTCP~\cite{alsabah2013pctcp},
which uses a dedicated TCP connection between each relay for every circuit,
has the potential to be actually deployed in Tor.
While PCTCP provides some improvements, e.g., in fairness,
it still does not provide sufficient congestion control.
Other approaches often require changes to the network infrastructure
and are therefore not directly applicable.

The problem of congestion in networks has also been studied extensively from a control theoretical perspective in the past.
Previous works include classic linear control~\cite{Mascolo1999} including PID~\cite{Yanfie2003} and state-feedback LQR control~\cite{Azuma2006}. 
It is well understood 
that delay is among the main challenges of controlling the network.
More recently, especially optimization-based methods have been applied to the problem with promising results~\cite{He2007,Mota2012}. 
Model predictive control (MPC), as applied in~\cite{Mota2012}, 
is an advanced control technique that can deal with non-linear systems and
explicitly take constraints into consideration. 
Its predictive control action is particularly suited for systems with significant delay.
Furthermore, MPC has received significant attention as a method for distributed control~\cite{Negenborn2014, CHRISTOFIDES2013}, 
where local controllers interact to jointly control an interconnected system. 
Distributed MPC is often applied to systems with a complex network character,
such as transportation systems~\cite{Dunbar2012},
energy management~\cite{Patel2016}
or process industry applications~\cite{CHRISTOFIDES2013}, 
where a centralized solution is prohibitive due to the size of the system
or privacy concerns.
In order to obtain global properties, the local action is often coordinated by exchanging information about predicted future behavior~\cite{Negenborn2014}. 

In this paper, we develop \emph{PredicTor} (see Figure~\ref{fig:paper_overview}),
a distributed MPC congestion control algorithm for the Tor network.
The distributed design is imperative, to allow scaling the network and most importantly, maintain anonymity of the users. 
In contrast to the current behavior of Tor,
PredicTor avoids congestion by generating \emph{backpressure}.
This denotes the strategy of propagating congestion back to the original sender
instead of accumulating it within the network,
and it is achieved via the information exchange
of the proposed distributed MPC.
Furthermore, PredicTor is designed with a focus on fair allocation of resources.
While optimization-based rate allocation is a well researched topic,
with equivalent formulations for TCP and other methods~\cite{He2007},
we introduce in this work a novel optimization-based
\emph{max-min fairness} formulation.
To the best of our knowledge, PredicTor is the only distributed MPC approach
to tackle the previously mentioned congestion and fairness challenges of the Tor network.
While distributed MPC has been applied to regular computer networks before~\cite{Mota2012},
our approach explicitly considers fairness and the applicability to a real network.

With a real application in mind, we design PredicTor as a modification of the Tor protocol.
For evaluation purposes, we build a prototype based on the ns-3 network simulator
and its extension nstor~\cite{backtap}.
Our results indicate that PredicTor can clearly reduce the latency of data transmission
as well as the load on the network.
As an additional contribution, our implementation of PredicTor and the required adaptations to nstor are available as an open source software project.%
\footnote{\texttt{https://github.com/cdoepmann/predictor}}

The remainder of this paper is structured as follows.
In Section~\ref{sec:tor_structure}, we present the structure of the Tor network, including related terminology and mathematical notation. 
Our main contribution is presented in Section~\ref{sec:mpc_formulation}.
First, we present and proof Theorem~\ref{theo:max_min_optim} 
which is an optimization-based method to obtain max-min fairness. 
We then discuss the dynamic system model and our distributed MPC concept
before we present the full PredicTor formulation.
In Section~\ref{sec:results}, we showcase the performance of PredicTor in a ns-3 network simulation study of an exemplary Tor topology.

\section{Structure of the Tor network}
\label{sec:tor_structure}
In order to achieve anonymity,
the Tor network~\cite{tor} provides a set of relay nodes.
These relays are used by clients to tunnel their communication through the network.
The established paths are commonly referred to as circuits
and carry equally-sized packets.
Anonymity in Tor is achieved by the fact that
a server cannot tell where client data originates from,
since the server only sees the last relay in the circuit.
The Tor relays form a so-called \emph{overlay network},
a computer network that operates on top of the public Internet.
The necessary resources (servers and bandwidth) are provided by volunteers
and are not subject to any central authority.
An exemplary Tor topography is depicted in Figure~\ref{fig:setup_example_network}.
It contains three circuits that share a set of six relays.
One of the relays was (randomly) chosen by all three circuits
and thus constitutes a possible bottleneck.
This scenario is prototypical for commonly observed behavior in the Tor network.
Note that circuits generally carry data bidirectionally.
For simplicity, we only consider one direction in this paper;
the other direction can be realized completely analogously.

\begin{figure}
	\footnotesize
	\centering
	\def\svgwidth{0.9\linewidth}
	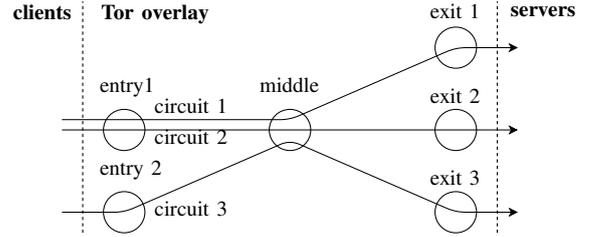
	\caption{Exemplary topology for a Tor network.}
	\label{fig:setup_example_network}
\end{figure}

Formally, we introduce Tor as an overlay network graph $G(N,E)$ where $N$ denotes the set of nodes and $E$ the set of overlay links. The network has a total of $|N|=n$ nodes and $|E|=e$ connections. 
We denote the set of Tor circuits $P$ with $i \in P$ being the i-th circuit of the set of cardinality $|P|=p$.
$P_{\alpha} \in P$ denotes the subset of circuits traversing node $\alpha \in N$.
Generally, we refer to circuits with Roman letters and to nodes with Greek letters.
When considering the network at the circuit level, we denote with $r_i$ the data rate (in packets per second) at which a circuit~$i$ is transferring information.
Furthermore, each node $\alpha \in N$ of the overlay network has a limited capacity $C_{\alpha}$, since overlay connections share the same physical connection.
\begin{definition}
	\label{def:feasible_rate}
	A rate vector $r=[r_1, r_2, \dots, r_p]$ is feasible if:
	\begin{align}
	\forall i \in P:& \quad    0\leq r_i \quad  \text{and}\\
	\forall \alpha \in N:& \quad    \sum_{i\in P_\alpha} r_i \leq C_{\alpha}.
	\end{align}
	We denote $R_f$ the set of feasible rate vectors.
\end{definition}
Each node $\alpha \in N$ can receive, store and send data from each circuit~$i \in P_{\alpha}$.
We denote $s_{\alpha, i} $ the circuit queue (storage in number of packets) in node $\alpha$ for circuit~$i$ and the vector with all queues for each circuit in node $\alpha$ as $s_{\alpha} \in \mathbb{N}^{|P_{\alpha}|}$.
Congestion of the network results from high values of these circuit queues and can be quantified with the data backlog.
\begin{definition}
    The data backlog ($b$) of a network $G(N,E)$ is computed for all nodes $\alpha \in N$ and all circuits $i\in P$ as:
    \begin{equation}
        b = \sum_{\alpha \in N} \sum_{i \in P} s_{\alpha,i}.
    \end{equation}
\end{definition}

\vspace{.5em}

\section{PredicTor}\label{sec:mpc_formulation}
The proposed predictive controller for the Tor network (PredicTor) is developed with three objectives in mind:
Primarily, we are aiming to avoid congestion of the network by limiting the data backlog of circuits.
Secondly, we seek to fully utilize the available resources of the network,
and lastly, we require a fair allocation of these resources.

This section starts by deriving an optimization-based method to obtain global fairness in Subsection~\ref{ssec:optim_fairness}.
We also show that this formulation will satisfy our second objective and utilize the available resources.
Congestion control, our primary objective, 
can only be achieved by exchanging the predicted action between connected nodes. 
The concept for exchanging information is presented in Subsection~\ref{ssec:feedback_distr_mpc}.
We then present the full optimal control problem~(OCP) in Subsection~\ref{ssec:ocp}.
Finally, we discuss the interaction of PredicTor and a Tor relay in Subsection~\ref{ssec:controller_tor_interact}.

\subsection{Optimization-based
fairness}\label{ssec:optim_fairness}
In the following we present an optimization-based method (Theorem~\ref{theo:max_min_optim}) to achieve max-min fairness.
We consider for the derivation the global rate $r_i$ for circuit $i \in P$.
\begin{definition}
\label{def:max_min_fair}
A feasible rate vector $r^f\in R_f$ is called max-min fair, if for all circuits $i\in P$ and for all other feasible rates $\bar{r}\in R_f$ it holds that:
\begin{equation}
	\begin{gathered}
	     \bar{r}_i \geq r_i^f \Rightarrow \exists \, j\in P:  r_{j}^f \leq r_i^f \land \bar{r}_j \leq r_j^f.
	\end{gathered}
\end{equation}
This definition means that if a rate $r^f$ is max-min fair, any other feasible rate that increases the rate for the favored circuit~$i$ comes at the cost of reducing the rate for the disadvantaged circuit~$j$, which is already smaller than the rate of circuit~$i$. 
\end{definition}
\begin{definition}
	\label{def:bottleneck}
	For a circuit $i\in P_{\alpha}$ and a rate vector $r$, we denote node $\alpha \in N$ a bottleneck, if:
	\begin{equation}
	\sum_{i \in P_{\alpha}} r_i = C_{\alpha}, \quad \forall j \in P_{\alpha}: \ r_i \geq r_j
	\end{equation}
\end{definition}
\begin{lemma}
	\label{lemma:min_max_bottleneck}
	Let $r^f$ be a max-min fair rate vector. Each circuit $i \in P$ has exactly one bottleneck.
	This bottleneck is the global rate-limiting factor of the circuit under stationary conditions.
\end{lemma}
\begin{proof}
	The proof is shown in~\cite{bertsekas1992data}.
\end{proof}
We can now state one of the main contributions of this work: how to obtain a max-min fair rate $r$ by solving a convex optimization problem.
\begin{theorem}
\label{theo:max_min_optim}
An overlay network achieves max-min fairness with rate  $r = r^{\text{max}} - \Delta r$ as the optimal solution of:
\begin{equation}
    \label{eq:opt_max_min_fairness}
    \begin{aligned}
       c= \min_{\Delta r} \sum_{i \in P}& \Delta r_i^2\\
        \text{subject to:}\quad
         r^{\text{max}} - \Delta r&\in R_f,\\
        0\leq \Delta r &\leq r^{\text{max}}
    \end{aligned}
\end{equation}
where $r^{\text{max}}$ is an arbitrary upper limit with
${\Delta r^{\text{max}}\geq\max(C_1, C_2, \dots, C_n)}$.
\end{theorem}
\begin{proof}
Proof by contradiction. Assume that the optimal solution $r^*$ with optimal cost $c^*$ is not fair.
If $\forall i \in P$ it holds that $\ r_i^* \leq r^f_i$:
\begin{equation*}
	c^* = \sum_{p_i \in P_n} (\Delta r_i^*)^2 \geq \sum_{p_i \in P_n} (\Delta r_i^f)^2  = c^f
\end{equation*}
Since the fair rate $r^f$ is feasible, the assumed solution is not optimal.
On the other hand, if we favor circuit~$i$ with the rate $r_i^* \geq r_i^f$, then, by Definition~\ref{def:max_min_fair} we disadvantage circuit~$j$ with rate $r_j$:
\begin{equation*}
	\exists p_j:  r_{j}^f \leq r_i^f \land r_j^* \leq r_j^f.
\end{equation*}
This means that $\Delta r_j^f \geq \Delta r_i^f$ and $\Delta r_j^* \geq \Delta r_j^f$.
We denote the magnitude of the disadvantage given to circuit~$j$ by $\Delta r_j^*-\Delta r_j^f = m$.
Considering Definition~\ref{def:bottleneck} and Lemma~\ref{lemma:min_max_bottleneck}, we note that a disadvantage of magnitude $m$ for circuit~$j$ is an upper bound for the possible advantage that can be given to circuit~$i$:
\begin{align*}
	\Delta r_j^*- \Delta r_j^f &\geq \Delta r_i^f-\Delta r_i^*, \\
	\intertext{We can now substitute $m$:}
	 m &\geq \Delta r_i^f-\Delta r_i^*.
\end{align*}
Rearranging the terms above leads to:
\begin{align*}
	\Delta r_i^* &\geq \Delta r_i^f - m.
\end{align*}
Together with $\Delta r_j^*-\Delta r_j^f = m$ we can write the difference between the optimal cost and a a max-min fair cost as:
\begin{align*}
	c^* - c^f &= (\Delta r_i^*)^2 + (\Delta r_j^*)^2 - (\Delta r_i^f)^2 - (\Delta r_j^f)^2 \\ & \geq (\Delta r_j^f + m)^2 + 	(\Delta r_i^f - m)^2 - (\Delta r_i^f)^2 - (\Delta r_j^f)^2 \\
	& \geq 2\Delta r_j^f m - 2\Delta r_i^fm + 2m^2 \\
	& \geq 2m(\Delta r_j^f - \Delta r_i^f) + 2m^2 \geq 0,
\end{align*}
where the first equality is given by the fact that only the circuits i and j are different for the optimal and fair solutions. 
The last inequality holds because circuit~$j$ is already disadvantaged  ($\Delta r_j^f \geq \Delta r_i^f$). The last inequality implies that the fair cost is smaller or equal than the optimal cost, 
which is a contradiction and proofs that the optimal solution of problem~\ref{eq:opt_max_min_fairness} yields the max-min fair rate vector $r$.
\end{proof}

\subsection{Distributed MPC}\label{ssec:feedback_distr_mpc}
PredicTor is a distributed MPC approach, where an optimal control problem is repeatedly solved at each node $\alpha \in N$ of the network, to obtain local decisions regarding incoming and outgoing rates. 
To achieve global fairness while maintaining constrained backlogs,
adjacent nodes need to exchange information about their predicted future actions.

Predictions are obtained on the basis of a dynamic model, 
for which we denote $s_{\alpha,i}^k$ the queue of a circuit~$i$ in node~$\alpha$ and at time step $k$. The dynamic model equation can be written as:
\begin{align}
\label{eq:mpc_model}
s_{\alpha,i}^{k+1} &= s_{\alpha,i}^k + \Delta t(r_{\text{in}, \alpha,i}^k - r_{\text{out}, \alpha,i}^k),
\end{align}
where $\Delta t$ denotes the sampling time.
We differentiate between incoming ($r_{\text{in},\alpha}$) and outgoing ($r_{\text{out},\alpha}$) rate, which can vary from
the overall rate for circuit~$i$, due to local storage terms.

For the interaction of multiple nodes, we denote  $\alpha \in N$ the currently considered node, with connections to predecessor ($\beta$) and successor ($\gamma$) nodes.
For the current node $\alpha$ it is irrelevant whether the incoming data comes from several nodes or only from a single node. For this reason, we assume that all incoming data for all different circuits comes from a single predecessor node $\beta$.

To further facilitate the statement of the optimization problem as well as the investigation of the proposed method, we assume in the following that all connections in $E$ of the network $G(N,E)$
experience a constant delay which is equivalent to the timestep ($\Delta t$) of the control problem.
We want to emphasize that the proposed algorithm is not restricted to that case
and can be easily adapted for the case of varying delays.

Information is exchanged at each MPC time step, such that node $\alpha$ receives messages with predicted trajectories of all connected nodes.
The interaction of node $\alpha$ with its  incoming node $\beta$ is shown in Figure~\ref{fig:feedback_exchange}.
\begin{figure}
	\footnotesize
	\centering
	\def\svgwidth{0.8\linewidth}
	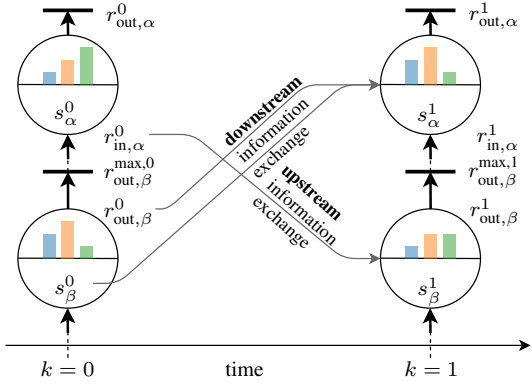
	\caption{Information exchange between nodes $n_{\alpha}$ and $n_{\beta}$.}
	\label{fig:feedback_exchange}
\end{figure}
We differentiate between upstream and downstream information exchange.
This is important because downstream messages travel with the data and latency between nodes can be omitted.
Upstream messages, on the other hand, are traveling opposed to the data and are therefore delayed.
The downstream information contains the predicted outgoing rate $\mathbf{r}_{\text{out},\beta}^k$ and the predicted circuit queue $\mathbf{s}_{\beta}^k$ from the predecessor node. 
For clarity, we denote trajectories with bold letters, such that 
$\mathbf{r}_{\text{out},\beta}^k=[r_{\text{out},\beta}^k, r_{\text{out},\beta}^{k+1},\dots, r_{\text{out},\beta}^{k+N_{\text{horz}}}]$, where
$N_{\text{horz}}$ is the MPC prediction horizon.

In the upstream direction, node $\alpha$ sends information about the predicted incoming rate $\mathbf{r}_{\text{in},\alpha}^k$ and receives $\mathbf{r}_{\text{in},\gamma}^k$ from its successor.
The exchanged information is affecting the local outgoing rate as:
\begin{equation}
\label{eq:mpc_pred_r_out_max}
\mathbf{r}_{\text{out},\alpha}^{k} \leq
\mathbf{r}_{\text{out},\alpha}^{\text{max},k}
=\mathbf{r}_{\text{in},\gamma}^{k-1},
\end{equation}
where we consider $\mathbf{r}_{\text{in},\gamma}^{k-1}$ from the previous time step since upstream information is delayed.
The formulation in \eqref{eq:mpc_pred_r_out_max} means that our successor can directly limit our outgoing rate.
Furthermore, the current node $\alpha$ estimates the queue size of its predecessor node $\beta$ based on $\mathbf{r}_{\text{out},\beta}^k$, $\mathbf{s}_{\beta}^k$ as well as $\mathbf{r}_{\text{in},\alpha}^k$,
where the latter is an optimization variable. 
The estimated queue size of node $\beta$ from the perspective of node $\alpha$ at time $k$ is denoted as $\tilde{s}_{\alpha|\beta}^{k}$.
We introduce the state variable $\Delta s_{\alpha|\beta}$
which allows us to formulate an expression for $\tilde{s}_{\alpha|\beta}^{k}$:
\begin{subequations}
	\label{eq:mpc_pred_predecessor}
	\begin{align}
	\tilde{s}_{\alpha|\beta}^{k} &= s_{\beta}^k - \Delta s_{\alpha|\beta}^{k}, \\
	\Delta s_{\alpha|\beta}^{k+1} &= \Delta s_{\alpha|\beta}^{k} +\Delta t( r_{\text{in},\alpha}^k - r_{\text{out},\beta}^k)
	\label{eq:mpc_update_eq_delta_s}.
	\end{align}
\end{subequations}%
Equation~\eqref{eq:mpc_pred_predecessor} states that any value $r_{\text{in},\alpha}^k\neq r_{\text{out},\beta}^k$ will adjust the predicted circuit queue at the predecessor node.
This plays an important role for the distributed MPC formulation, as incoming rates $\mathbf{r}_{\text{in},\alpha}$ can explicitly consider the availability of data, by introducing: $\tilde{\mathbf{r}}_{\alpha|\beta}^{k}\geq 0$.
\vspace{4em}

\subsection{Optimization problem}\label{ssec:ocp}
We propose the following OCP for congestion control with fairness formulation for node $\alpha$, predecessor node $\beta$ and successor node $\gamma$.
\begin{small}
\begin{subequations}
\label{eq:mpc_full_optim}
\begin{align}
\min_{
\Delta \mathbf{r}_{\text{out},\alpha},\Delta \mathbf{r}_{\text{in},\alpha}, \mathbf{s}_{\alpha}, \Delta \mathbf{s}_{\alpha|\beta}
}&
\sum_{k=0}^{N_{\text{horz}}} d^k \left((\Delta r_{\text{in},\alpha}^k)^2+ (\Delta r_{\text{out},\alpha}^k)^2\right)
\label{eq:mpc_full_optim_01}\\
\text{subject to }&: \nonumber\\
s_{\alpha}^{k+1} &= s_{\alpha}^k + \Delta t(r_{\text{in},\alpha}^k - r_{\text{out}, \alpha}^k),\\
\Delta s_{\alpha|\beta}^{k+1} &= \Delta s_{\alpha|\beta}^{k} +\Delta t( r_{\text{in},\alpha}^k - r_{\text{out},\beta}^k),\\
0&\leq r^{\text{max}} - \Delta r_{\text{in},\alpha}^k,\\
0&\leq r^{\text{max}} - \Delta r_{\text{out},\alpha}^k \leq r_{\text{out},\alpha}^{\text{max},k},\\
\sum_{i\in P_{\alpha}}  &\left(r^{\text{max}} - \Delta r_{\text{in},\alpha}^k \right) \leq C_{\alpha}^{\text{in}}\\
\sum_{i\in P_{\alpha}}  & \left(r^{\text{max}} - \Delta r_{\text{out},\alpha}^k \right)\leq C_{\alpha}^{\text{out}}\\
0 &\leq s_{\alpha}^k \leq s_{\alpha}^{\text{max}},
\label{eq:mpc_full_optim_sc_max}\\
0 &\leq  s_{\beta}^k - \Delta s_{\alpha|\beta}^{k}\\
s_{\alpha}^0 &= s_{\alpha}^{\text{init}},\ \Delta s_{\alpha|\beta}^{0} = 0.\\
&\forall k = 0 ,\dots, N_{\text{horz}} \nonumber
\end{align}
\end{subequations}
\end{small}%
To solve \eqref{eq:mpc_full_optim_01}, the predicted trajectories of adjacent nodes $\mathbf{r}_{\text{in},\gamma}^{k-1}$, 
$\mathbf{r}_{\text{out},\beta}^k$ and $\mathbf{s}_{\beta}^k$,
as well as the current size of the circuit queue in the current
node $s_{\alpha}^{\text{init}}$ have to be supplied.
Note that according to \eqref{eq:mpc_pred_r_out_max}, 
we set $\mathbf{r}_{\text{out},\alpha}^{\text{max},k}
=\mathbf{r}_{\text{in},\gamma}^{k-1}$.

The objective in~\eqref{eq:mpc_full_optim_01} is motivated by the presented Theorem~\ref{theo:max_min_optim} but with some important adaptations.
Most notably, we introduced $\Delta r$ variables for both the incoming and outgoing rates.
Introducing the control variable $\Delta r_{\text{in},\alpha}$ allows to control the incoming rate.
This  is of significant importance for the desired congestion control as it induces \textit{backpressure} and data will be stopped from entering the network if it cannot be forwarded.
The quadratic term in $\Delta r_{\text{out},\alpha}$ ensures that the circuit queue is emptied even if there are no new packets entering the node.

The objective in~\eqref{eq:mpc_full_optim_01} is further modified by introducing a discount factor ($d$).
This is necessary because naively implementing our presented fairness formulation also results in fairness along the prediction horizon,
where it is always preferable to increase the rate of the smallest element in a sequence for a given circuit.
In practice, however, we want to send and receive as soon as possible as long as \emph{instantaneous} fairness is achieved.
In Appendix A we present a guideline on how to choose an upper bound for $d$ to obtain the desired behavior.

We implement PredicTor based on CasADi \cite{Andersson2018} in combination with IPOPT \cite{Andreas2006} and MA27\footnote{
HSL. A collection of Fortran codes for large scale scientific computation. \texttt{http://www.hsl.rl.ac.uk/}} linear solver
for fast state-of-the-art optimization.

\begin{figure*}
	\centering
	\includegraphics[width=1\linewidth]{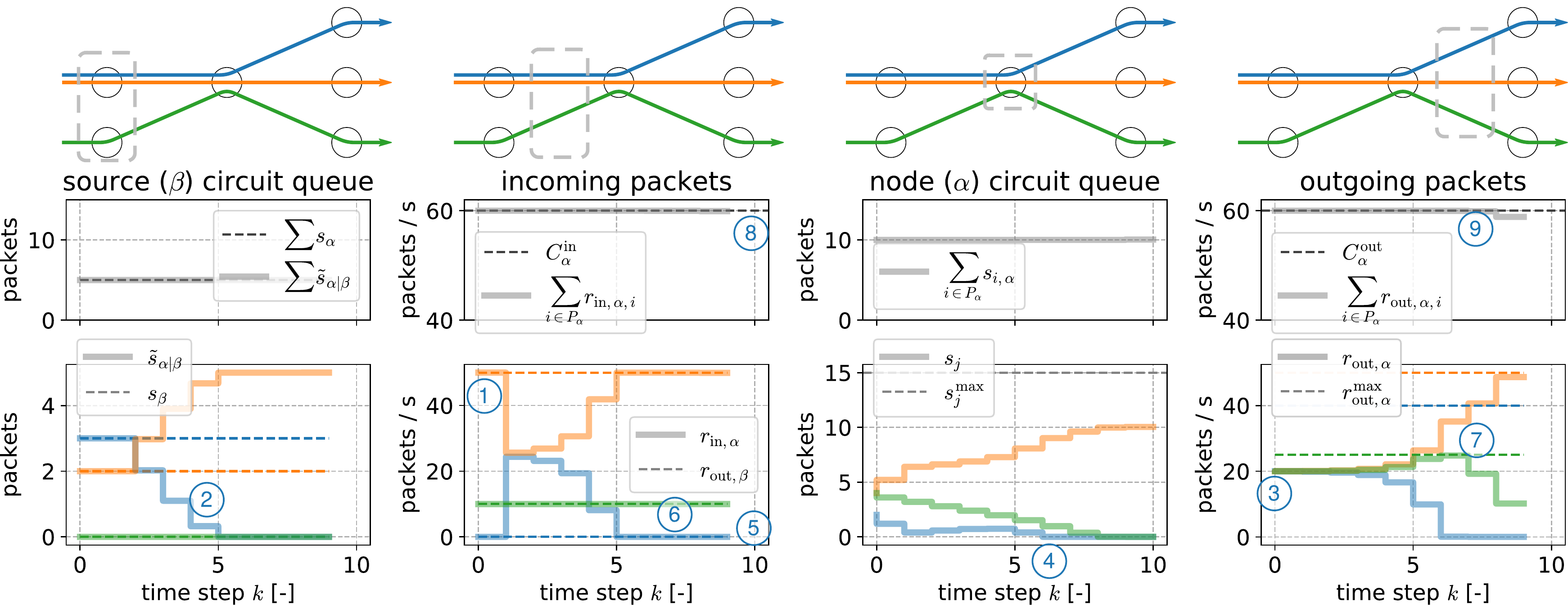}
	\caption{Open-loop MPC prediction for  the central node as shown in Tor topology from Figure~\ref{fig:setup_example_network}. The circled numbers are referred to in the results discussion.}
	\label{fig: open_loop_pred}
\end{figure*}%

\subsection{Interaction of controller and network}\label{ssec:controller_tor_interact}

The proposed controller is implemented on the application layer of each node in the Tor network. 
At each timestep, problem~\eqref{eq:mpc_full_optim} is solved with the most recent measurement of the circuit queue $s_{\alpha}^{\text{init}}$ of the current node $\alpha \in N$ and with the received information from adjacent nodes.
The optimal solution of \eqref{eq:mpc_full_optim} is converted to trajectories of incoming ($\mathbf{r}_{\text{in},\alpha}$) and outgoing ($\mathbf{r}_{\text{out},\alpha}$) rates,
where the first element of $\mathbf{r}_{\text{out},\alpha}$
is used to control at which rate data is sent.
In particular, we employ a token-bucket method~\cite{bertsekas1992data}
to shape the outgoing traffic.
Note that we are controlling the data rates 
on a per-circuit basis, which is similar to the aforementioned PCTCP~\cite{alsabah2013pctcp}.

In order to exchange the trajectories between relays,
we extend the Tor protocol with respective control messages.
Entry (exit) nodes do not have a
predecessor (successor) to exchange data.
In this case, we provide reasonable, synthetic trajectories to bootstrap the data transfer,
and behave accordingly.
For example, the first node in a circuit reads data from its source
according to its computed incoming rate.

While PredicTor is generally agnostic to the underlying transport protocol,
we implement it using TCP as a reliability mechanism,
to avoid packet loss and packet reordering.

\section{Results}\label{sec:results}
Evaluating the proposed controller on the live Tor network
is neither feasible nor responsible, due to its sensitive nature.
Instead, the performance of PredicTor is investigated in simulation studies.
To evaluate the performance of PredicTor,
we use ns-3, a discrete-event network simulator that offers
a safe simulation environment.
It achieves a high degree of realism
by emulating the network down to the physical layer,
including queueing effects, potential packet loss, and other network effects.

For the evaluation,
we focus on several core metrics that are relevant in this context:
The amount of data transferred within in a given time span gives an indication
of how well the available resources are utilized.
Comparing these values between circuits constitutes a measure of fairness.
On the other hand, the byte-wise latency
between data entering and leaving the network
is important to characterize the applicability of the system for end users.
Latency is strongly influenced by the size of queues within in the network.
We therefore also consider the backlog,
which constitutes a metric
for the overall load of the network.

The results presented in the following are obtained
with a discount factor as shown in \eqref{eq:exp_discount_factor}, 
where we choose $d_0=\frac{1}{3}$, as discussed in Appendix A.

\subsection{Open-loop prediction}
We begin this section by investigating the decision-making process of the proposed controller in~\eqref{eq:mpc_full_optim} by studying an open-loop prediction.
This allows to highlight several interesting aspects of the behavior of PredicTor,
before we present the closed-loop distributed application.
In Figure~\ref{fig: open_loop_pred}, we showcase a result of~\eqref{eq:mpc_full_optim} for the central node of the presented Tor topology in Figure~\ref{fig:setup_example_network}.
Again, we denote $\alpha$ the current node, $\beta$ its predecessor and $\gamma$ its successor.

We created a synthetic scenario, consisting of trajectories $\mathbf{r}_{\text{out},\alpha}^{\text{max}}$,
$\mathbf{r}_{\text{out},\beta}$, $\mathbf{s}_{\beta}$, as well as the initial circuit queues ($s_{\alpha}^{\text{init}}$).
For the incoming connections, the proposed controller can now determine the optimal trajectory $\mathbf{r}_{\text{in},\alpha}$.
Since we need to consider the time delay (upstream information exchange),
changes to the outgoing rate $\mathbf{r}_{\text{out},\beta}$ are not immediately possible~\circled{1}.
The effect of the change $\mathbf{r}_{\text{in},\alpha}-\mathbf{r}_{\text{out},\beta}$  shows itself in the corrected prediction of the source circuit queue.
We see, for example, how $\mathbf{r}_{\text{in},\alpha,1}$ for circuit 1 is chosen such that the queue for that circuit is emptied at the source node~\circled{2}.

For the outgoing connections, we obtain the optimal trajectory ($\mathbf{r}_{\text{out},\alpha}$).
We notice that at the beginning of the horizon,
all circuits are assigned the same fair rate~\circled{3}.
Over time, the controller first reduces the rate for circuit 1 ($\mathbf{r}_{\text{out},\alpha,1}$), as its circuit queue is emptying~\circled{4} and the incoming rate is also vanishing~\circled{5}.
Afterwards, the controller reduces the rate for circuit 2 ($\mathbf{r}_{\text{out},\alpha,2}$)  to meet exactly the incoming rate ($\mathbf{r}_{\text{in},\alpha,2}$)~\circled{6}.
We can also see how individual constraints for $\mathbf{r}_{\text{out},\alpha}^{\text{max}}$ are obeyed~\circled{7} as well as capacity constraints for incoming~\circled{8} and outgoing rates~\circled{9}.
 
\subsection{Comparison}
In this section, we showcase the closed-loop behavior of PredicTor for two different scenarios and the presented Tor topology from Figure~\ref{fig:setup_example_network}.
We compare the performance of PredicTor with Tor (current standard) as well as PCTCP.
Closed-loop means that at each MPC step we are updating the current state of the system from measurements obtained with ns-3.
Furthermore, we receive updated information from all adjacent nodes.
Two scenarios are investigated. In scenario 1, circuits 1-3 start sending at the beginning of the simulation window.
Circuits 1 and 3 have an infinite source of packets to forward, whereas circuit 2 stops and restarts twice during the simulated window.
In  scenario 2, all circuits have an infinite source of packets.
This scenario is considered for the fairness evaluation
because it better approximates stationary behavior.

\begin{figure}
	\small
	\centering
	\def\svgwidth{1\linewidth}
	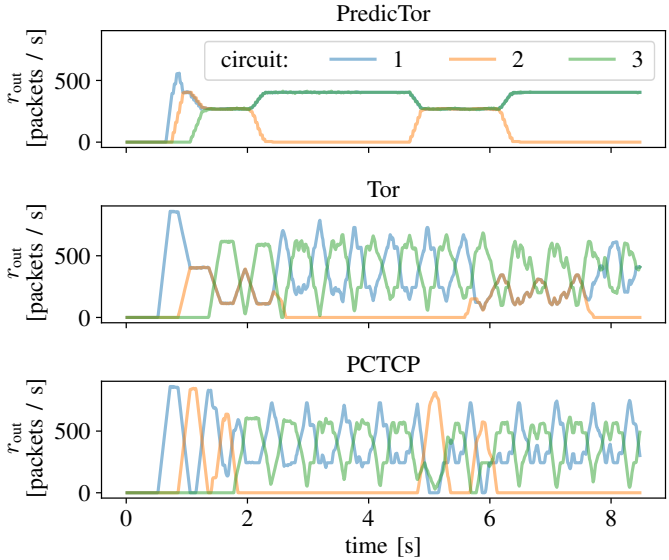
	\caption{Comparison of outgoing rates ($r_{\text{out}}$) for  the central node as shown in the Tor topology from Figure~\ref{fig:setup_example_network} for scenario 1. 
		Rates for all methods are calculated from ns-3 simulation results as the number of packets that are forwarded within one sampling time step ($\Delta t=0.04\si{s}$).}
	\label{fig: compare_flavour_r_out}
\end{figure}

In Figure~\ref{fig: compare_flavour_r_out}, we display the outgoing rates ($r_{\text{out}}$) of the middle node for scenario 1 over the course of the simulation time.
Note that all rates are obtained from the ns-3 simulation.
PredicTor shows a desirable behavior with constant, sustainable rates and smooth transitions when circuit 2 stops and restarts.
Fair behavior can be observed in these transitions: all circuits share the same rate during activity and circuit 1 and 3 are allocated the same, higher rate when circuit 2 stops sending. The sum of all rates is visibly constant over time, thus PredicTor is fully utlilizing the available resources.
On the other hand, Tor and the PCTCP adaptation show  erratic, oscillatory behavior where bursts are followed by very low rates. 
Fairness cannot be assessed visually for Tor and PCTCP, 
which is why we quantify it later in Table~\ref{tab:flavor_comparison}.

We further compare PredicTor, Tor, and PCTCP in Figure~\ref{fig: compare_flavour_backlog} and \ref{fig: compare_flavour_latency},
where we display the backlog and latency. 
PredicTor succeeds at its primary goal of sustaining a manageable backlog, 
especially compared to Tor's and PCTCP's approach.
The importance of this effective congestion control becomes apparent in Figure~\ref{fig: compare_flavour_latency},
where we compare histograms for the latencies of received packets.
PredicTor significantly improves on Tor and PCTCP with an average latency of 106~ms,
in contrast to 558~ms and 624~ms, where a theoretical minimum of 80~ms is possible.

\begin{figure}
	\small
	\centering
	\def\svgwidth{1\linewidth}
	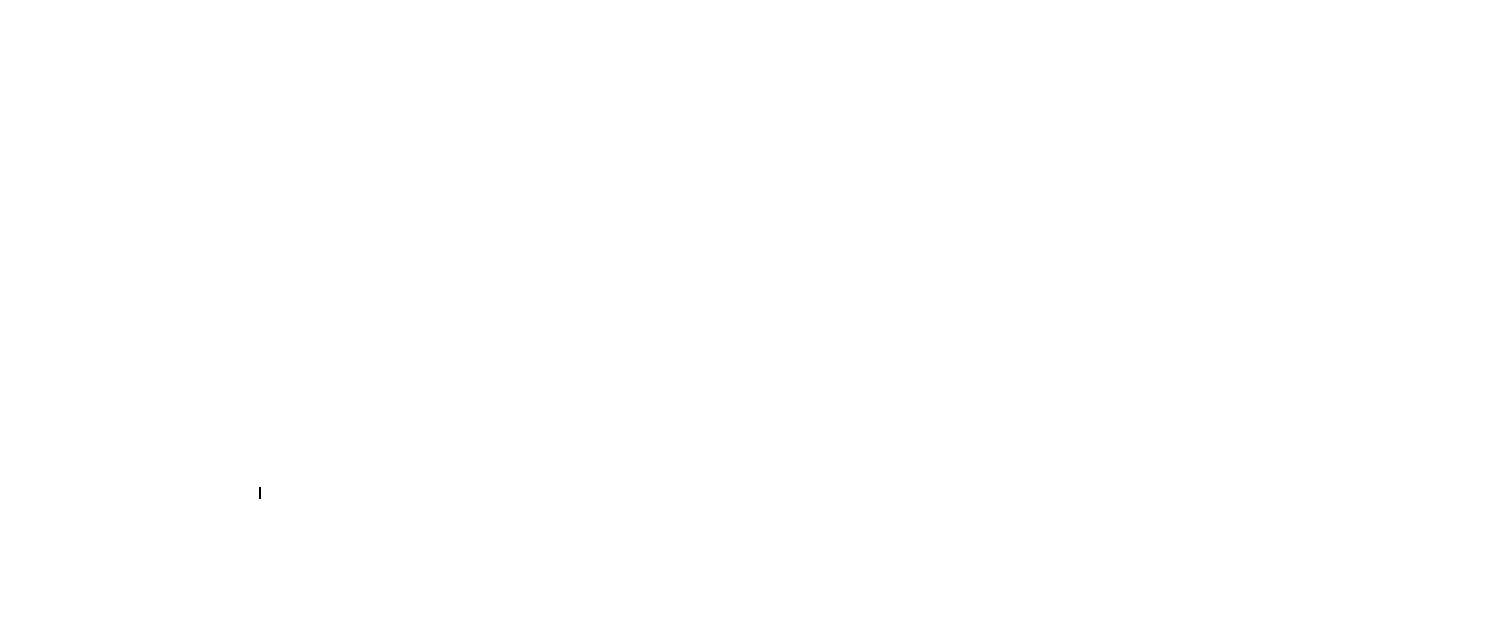
	\caption{Comparison of total backlog (packets in the network) in the Tor topology from Figure~\ref{fig:setup_example_network} for scenario 1. 
	Backlogs are calculated from ns-3 simulation results.}
	\label{fig: compare_flavour_backlog}
	\vspace{0.2cm}
	\def\svgwidth{1\linewidth}
	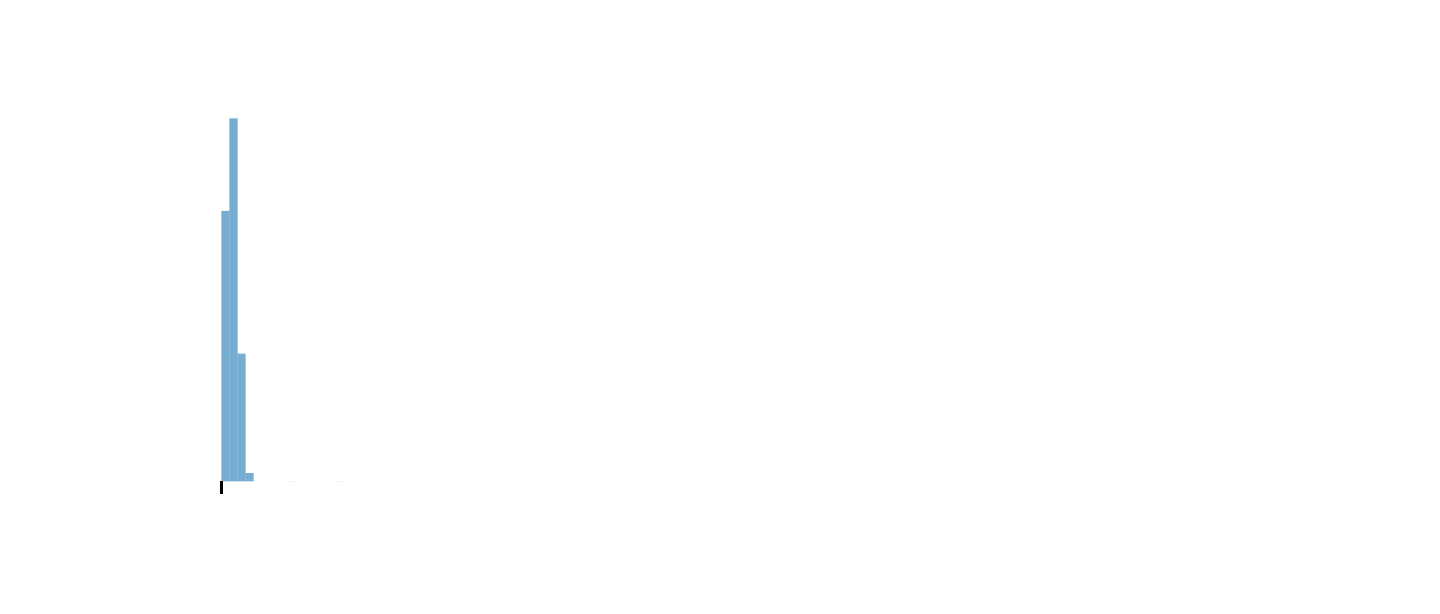
	\caption{Histogram of latency for received packets in the Tor topology from Figure~\ref{fig:setup_example_network} for scenario 1. 
	Latencies are calculated from ns-3 simulation results
	and cumulated for all three circuits.}
	\label{fig: compare_flavour_latency}
\end{figure}%

\begin{table}
	\caption{Comparison of latency and transferred data.}
	\renewcommand{\arraystretch}{1.3}
	\begin{tabular}{lllllll}
		\toprule
		& \multicolumn{3}{l}{\textbf{mean latency}$^{1}$} & \multicolumn{3}{l}{\textbf{data transferred}} \\
		& \multicolumn{3}{l}{[ms]} & \multicolumn{3}{l}{[$\times 10^5$ packets]} \\
		circuit& PredicTor  & Tor  & PCTCP  & PredicTor   & Tor    & PCTCP  \\
		\midrule
		1     & 103        & 536  & 596     & 8.93        & 6.82   & 8.72    \\
		2     & 117        & 523  & 669     & 8.92        & 6.82   & 9.58    \\
		3     & 105        & 589  & 643     & 8.93        & 13.13  & 8.69    \\
		\midrule
		\textbf{Total} & 106        & 558  & 624     & 26.78       & 26.78  & 26.98\\
		\bottomrule
		\multicolumn{7}{r}{\footnotesize $^{1}$scenario 1,  $^{2}$scenario 2}
	\end{tabular}
	\label{tab:flavor_comparison}
\end{table}

In Table~\ref{tab:flavor_comparison}, we summarize and compare
the mean latency as well as the total number of received packets (throughput)
for each circuit.
Regarding throughput, the three methods perform similarly,
with the difference that only PredicTor achieves near perfect fairness. 
Tor clearly discriminates circuit 1 and 2 which share a connection, 
while PCTCP, as expected, manages to revise this effect to some extent.

While being clearly advantageous with respect to backlog and latency,
our proposed congestion controller comes at the cost of using network capacity for the exchange of messages for distributed MPC.
These messages are not included in the presented figures above.
We quantify their effect for the presented scenario 
and found that this overhead would reduce the throughput by 5.39~\%.

Furthermore, it is clear that PredicTor introduces significant complexity compared to the previous methods.
However, the optimization problem~\eqref{eq:mpc_full_optim} is convex, which guarantees a global solution in polynomial time.
For the given scenario, we obtain a solution in around 0.1~ms (laptop-grade CPU),
which is sufficient for a timestep of 40~ms.
The problem complexity (number of optimization variables and constraints) grows linearly with the number of circuits per node 
and it is therefore expected that also realistically large topologies can be tackled with the approach in real-time.

\section{Conclusion}
In this work, we have proposed a novel model predictive control formulation 
to tackle the challenge of congestion in Tor, whilst fairly allocating resources.
PredicTor is a distributed approach that relies on exchanging information 
in the form of predicted actions to all adjacent nodes.
We evaluate the proposed method in the state-of-the-art network simulator ns-3
and compare results to the present method of congestion control in Tor, 
as well as an adaption thereof (PCTCP).
PredicTor significantly outperforms both methods in terms of latency.
In our test scenario, we reduced the mean latency from 558~ms (Tor) and 624~ms (PCTCP)
to just 106~ms. 
PredicTor is at the same time superior in fairness and has a similar throughput.
The exchange of information slightly reduces this last figure by 5.39~\%,
which is found to be an acceptable trade-off.

\addtolength{\textheight}{-0cm}   %

\section*{APPENDIX}
\subsection{Discount Factor}
\label{app:discount_factor}
\begin{claim}\label{claim:discount_factor}
    Let the discount factor take the form:
    \begin{equation}
        \label{eq:exp_discount_factor}
        d^{k+1} = d_0 \cdot d^{k},
    \end{equation}
    with $d_0 \leq \frac{1}{3}$.
    Let $\Delta \tilde{\mathbf{r}}_{\text{out},\alpha,i}$, $\Delta \tilde{\mathbf{r}}_{\text{in},\alpha,i}$ be a feasible solution of \eqref{eq:mpc_full_optim}.
    Under the assumption that for circuit $i\in P_{\alpha}$:
    \begin{equation}
    \label{eq:discount_f_ass01}
        \Delta \tilde{r}_{\text{in},\alpha,i}^{k} \leq
        \Delta \tilde{r}_{\text{in},\alpha,i}^{k+1},
    \end{equation}
    and if it is feasible, the optimal solution $ \Delta \mathbf{r}_{\text{in},\alpha}^{*}$ of \eqref{eq:mpc_full_optim} 
    will have the following property:
    \begin{align*}
        \Delta r_{\text{in},\alpha,i}^{k^*} &= \Delta \tilde{r}_{\text{in},\alpha,i}^{k}-m^k,\\
        \Delta r_{\text{in},\alpha,i}^{{k+1}^*} &= \Delta \tilde{r}_{\text{in},\alpha,i}^{k+1}+m^k,
    \end{align*}
    where $0<m^k\leq \Delta \tilde{r}_{\text{in},\alpha,i}^{k}$.
    This means that it is optimal 
    to reduce the rate at time $k$ by magnitude $m^k$ whilst increasing the rate at $k$ by the same magnitude.
    The same holds for the outgoing rate ($\Delta \mathbf{r}_{\text{out},\alpha,i}$).
\end{claim}
\begin{proof}
We denote with $c_{i,k}(m)$ the cost of \eqref{eq:mpc_full_optim}, where for circuit~$i$ at time $k$ the rate $\Delta \tilde{r}_{\text{in},\alpha,i}^{k}$ has been decreased by magnitude $m^k$ at the cost of increasing $\Delta \tilde{r}_{\text{in},\alpha,i}^{k+1}$ at $k+1$ by the same magnitude.
We want to find a value of $d_0$ for
$0 < m^k\leq\Delta \mathbf{r}_{\text{in},\alpha,i}^{k}$, such that:
\begin{equation*}
    c_{i,k}(m)- c_{i,k}(0) \leq 0.
\end{equation*}
After subtracting all unchanged terms, we obtain:
\begin{equation*}
\begin{gathered}
d^k \left(\Delta \tilde{r}_{\text{in},\alpha,i}^k - m^k \right)^2
 + d^{k+1} \left(\Delta \tilde{r}_{\text{in},\alpha,i}^{k+1} + m^{k} \right)^2\\ - d^k \left(\Delta \tilde{r}_{\text{in},\alpha,i}^k \right)^2 - d^{k+1} \left(\Delta \tilde{r}_{\text{in},\alpha,i}^{k+1} \right)^2
 \leq 0.
 \end{gathered}
\end{equation*}
We consider \eqref{eq:exp_discount_factor} and expand the quadratic terms, such that:
\begin{equation*}
\begin{gathered}
-2\Delta \tilde{r}_{\text{in},\alpha,i}^k m^k+\left(m^k\right)^2
+ d_0 \left(2\Delta \tilde{r}_{\text{in},\alpha,i}^{k+1} m^k+ \left(m^k \right)^2 \right) \leq 0.
\end{gathered}
\end{equation*}
Due to \eqref{eq:discount_f_ass01}, the inequality still holds when substituting:
\begin{equation*}
\Delta \tilde{r}_{\text{in},\alpha,i}^{k+1} = \Delta \tilde{r}_{\text{in},\alpha,i}^{k}. 
\end{equation*}
Considering that $0 \leq m^k$, we can further simplify the inequality:
\begin{equation*}
-2\Delta \tilde{r}_{\text{in},\alpha,i}^k + m^k
+ d_0 \left(2\Delta \tilde{r}_{\text{in},\alpha,i}^k+ m^k \right) \leq 0.
\end{equation*}
The inequality still holds when substituting $\Delta \tilde{r}_{\text{in},\alpha,i}^k = m^k$, since $m^k \leq \Delta\tilde{r}_{\text{in},\alpha,i}^k$:
\begin{align}
d_0 \left(3m^k \right) &\leq 
m^k, \nonumber\\
d_0 &\leq \frac{1}{3}.
\end{align}
The proof is identical for the outgoing rate ($\Delta \mathbf{r}_{\text{out},\alpha,i}$).
\end{proof}

\bibliography{IEEEabrv,2020_cctaconf}

\end{document}